\documentclass{X}

\RendicontiPagina{1}{Mattia Scomparin}{Mattia Scomparin}

\usepackage{amsmath}

\usepackage{amsfonts}
\usepackage{slashed}
\usepackage{xcolor} 

\newcommand{\vp}{\varphi}
\newcommand{\dx}{\dot{x}}
\newcommand{\ddx}{\ddot{x}}
\newcommand{\ve}{\varepsilon}
\newcommand{\etb}{e^{\pm\!\int_t\!b d\tau}}
\newcommand{\etbb}{e^{\pm\!\int_t\!b_\pm d\tau}}
\newcommand{\esb}{e^{\pm\!\int_s\!b d\tau}}
\newcommand{\esbb}{e^{\pm\!\int_s\!b_\pm d\tau}}
\newcommand{\B}{\bar{B}}
\newcommand{\dt}{\boldsymbol{\cdot}}
\def\smallint{\begingroup\textstyle \int\endgroup}
\newcommand{\p}[1] {\partial_{#1}}
\newcommand{\bc}[1] {\bar{c}_{#1}}

\title{First integrals of nonlinear differential equations from nonlocal constants}
\author{Mattia Scomparin}
\dedicated{mattia.scompa@gmail.com}

\summary{
A new method to find first integrals of nonlinear differential equations in Jacobi-type form is presented. The basic idea of our approach is to use
one-parameter perturbed motions to find well-conceived nonlocal constants that are conserved along solutions.
By means of such nonlocal framework we derive a set of  theorems that we apply to look for the first integrals of some relevant cases, where moreover a solution is obtained.
Applications also include some equations of the Painlev\'{e}-Gambier classification.
}

\intesta

\begin{document}
\maketitle

\section{Introduction}
\label{sec:intro}

Many phenomena are mathematically described by second-order \emph{nonlinear} Ordinary Differential Equations (ODEs). Generally, nonlinear ODEs are representative of models with a complicated behaviour  and some of them emerge as \emph{Euler-Lagrange} equations
\begin{equation}\label{eq:E-L}
\big[\partial_{\dx} L\big(t,x(t),\dx(t)\big)\big]^{\dt}\!-\partial_{x} L\big(t,x(t),\dx(t)\big)=0\,,
\end{equation}
for an appropriate \emph{Lagrangian} function $L(t,x,\dx)$. In our notation $\partial$ stands for partial derivative, e.g. $\partial_{\dx}\!=\!\partial/\partial\dx$, whereas the upper dot means total derivative with respect to $t$.

Numerous methods are available in the literature to get exact solutions of equation \eqref{eq:E-L} and, among all, finding \textit{first integrals} allows to obtain solutions in the form of quadratures \cite{XXSTANI}.
By definition, a first integral is a smooth point-function $\mathcal{I}(t,x,\dx)$ that is constant throughout equation \eqref{eq:E-L} for all $t$, i.e. $\dot{\mathcal{I}}(t,x,\dx)=0$ along solutions. In the case of linear ODEs, a number of well defined methods to find first integrals exists. However, the same cannot be said for the nonlinear case, where their direct computation is in general an open problem. In this respect, the celebrated \textit{Noether's Theorem} establishes a relation between invariance proprieties of a Lagrangian function and first integrals \cite{Noeth}. Nevertheless, finding non-trival symmetries is not always simple for many instances. 

In favourable circumstances, by inspection of the Lagrangian taken into consideration, 
one can find particular cases for which special functions, called nonlocal constants, yield true first integrals without appealing to Noether symmetries. In particular, a \textit{nonlocal constant} is a function that is conserved by equation \eqref{eq:E-L} whose value depends not only on the state $(t,x(t),\dx(t))$, but also on the whole history of the motion from its beginning at $t_0$. In Ref. \cite{Gorni2021} Gorni and Zampieri provide  a theory to generate nonlocal constants for second order ODE Lagrangians using the concept of one-parameter perturbed motions:
\begin{theorem}\label{teo:nonloconst}
Let $t\rightarrow x(t)$ be a solution to the Euler-Lagrange equation \eqref{eq:E-L}, and let $x_\ve(t)$, $\ve\in\mathbb{R}$, be a smooth family of perturbed motions such that $x_0(t)=x(t)$. Then, the following function is constant:
\begin{equation}\label{eq:nonlocGZ}
I= \partial_{\dx} L\big(t,x(t),\dx(t)\big)\cdot \partial_\ve x_\ve(t)\big\rvert_{\ve=0}-\int_{t_0}^t 
\partial_\ve L\big(s,x_\ve(s),\dx_\ve(s)\big)\big\rvert_{\ve=0}\,ds\,.
\end{equation}
\end{theorem} 

This result has been applied by Gorni and Zampieri in Refs. \cite{Gorni2021,Gorni_2020,Gorni2023} to study some nonlinear mechanical systems, which include: (i) homogeneous potentials of degree $k = -2$, (ii) mechanical systems with viscous fluid resistance, (iii) mechanical system with hydraulic (quadratic) fluid resistance, (iv) conservative and dissipative Maxwell-Bloch equations of laser dynamics, and (v) geodesics for the Poincar\'e's half-plane.
Beside that, the author extended Theorem \ref{teo:nonloconst}  and related applications to the framework of higher-order theories \cite{10077_31215,GORNI2022100262} and relativistic scalar field theories \cite{SCMromp}.

A viable parametrization to study second-order non-linear ODEs is the so-called \emph{Jacobi-type} form (or, Jacobi equation) \cite{Jacobi1845,Nucci2008UsingAO,GUHA20103247}. It is a second-degree polynomial in $\dx$ like this 
\begin{equation}\label{eq:J}
\ddx+\tfrac{1}{2}\partial_x\vp(t,x)\,\dx^2+\partial_t\vp(t,x) \,\dx+B(t,x)=0\,,
\end{equation}
where the two free functions $\vp(t,x)$ and $B(t,x)$ are both analytical in $t$ and $x$. Interestingly, by taking  $\vp(t,x)$ and $B(t,x)$ appropriately, one can reproduce a great variety of nonlinear ODEs. Among all, we recall some of the second-order equations of the \emph{Painlev\'{e}-Gambier} classification (see, e.g. Refs. \cite{10.1007/BF02419020, 10.1007/BF02393211,Dambrosi2021}).

Painlev\'{e}-Gambier equations are  attracting much attention in last years, since some problems related to their solutions, the \emph{Painlev\'{e} transcendents}, are already under discussion  \cite{2019xdf,XX2022}. 
In particular, it is well known that the only movable singularities of these equations are poles.

In  Refs. \cite{GUHA20103247, GUHA2022} Guha \emph{et al.} studied a subset of such equations, showing that \emph{generalized Sundman transformations} can be used to obtain certain new first integrals and particular solutions. To do this, they started from \eqref{eq:J} and derived some conserved functions of the form $p_1(t,\dx,x)\,e^{\int_t\!p_2(\tau,x(\tau)) d\tau}$ for one Painlev\'{e}-Gambier equation at a time.   As may be clearly seen, such functions ``appear''  nonlocal due to the presence of an exponential factor whose argument is an integral of $p_2(t,x)$. 

Therefore, it is natural to suppose that a not yet investigated interplay between Theorem \ref{teo:nonloconst} and equation \eqref{eq:J} exists. The discovery of such relation is the main result of our work.
Remarkably, thanks to the Lagrangian-based approach underlying nonlocal constants like  \eqref{eq:nonlocGZ}, our  analysis allows to recover and generalize the results of Guha \emph{et al.} as single formulas  that have the advantage of being applicable to the whole parametrization \eqref{eq:J}.

As pointed before, Theorem \ref{teo:nonloconst} works with Euler-Lagrange equations only and  to proceed we need an appropriate Lagrangian for equation \eqref{eq:J}.
For our purposes comes to help a result obtained by Nucci and Tamizhmani in Ref. \cite{Nucci2008UsingAO}. In that article they present a method devised by Jacobi to derive Lagrangians of many second-order ODEs. Among the examples provided, they demonstrate that equation \eqref{eq:J}  can be derived as the Euler-Lagrange equation for a Lagrangian
\begin{equation}\label{eq:Nucci}
L(t,x,\dx)\equiv\tfrac{1}{2}e^{\vp(t,x)}\dx^2+\delta_1(t,x)\dx+\delta_2(t,x)\,,
\end{equation}
where the functions $\delta_{1}$ and $\delta_{2}$ satisfy the condition $\partial_t\delta_1-\partial_x\delta_2=e^\vp B$. Expression \eqref{eq:Nucci} completes the set of information necessary to operationalize Theorem \ref{teo:nonloconst} and proceed with our analysis.

The organization of the paper is as follows. For pedagogical reasons, in Section \ref{sec:aut} we deal with the simplest case of the autonomous Jacobi equation \eqref{eq:J}, or with $\p{t}\vp=\p{t}B=0$. We prove through Theorem \ref{teo:nonloconst} the existence of the true first integral \eqref{eq:I}, that we use in Subsection \ref{sec:oiu} to recover some Painlev\'{e}-Gambier constants and solutions obtained in Refs. \cite{GUHA20103247, GUHA2022}.
In Section \ref{eq:lowe} we extend the autonomous picture to the whole  Jacobi equation \eqref{eq:J} by distinguishing the two cases (i) $\p{t}\vp=0$, $\p{t}B\ne 0$ and (ii) $\p{t}\vp\ne0$, $\p{t}B\ne 0$. Here, our Theorem \ref{teo:xxffx} (Subsection \ref{eq:pt0}) and Theorem \ref{teo:xxkjkffx} (Subsection \ref{eq:lowe}) lead to nonlocal constants whose integrand depends on $t$ and $x$ only (and not $\dx$). 
In both cases, we finally investigate the effective potential of such theorems obtaining constants and solutions for some non-autonomous Painlev\'{e}-Gambier equations \cite{GUHA20103247, GUHA2022} and beyond.

The main results of this work suggest that Lagrangian-based nonlocal constants like \eqref{eq:nonlocGZ} have interesting applications to the study of nonlinear ODEs like \eqref{eq:J}. We believe that our approach can further be extended to, in principle, many other cases with little additional effort.


\section{Autonomous Jacobi equation}
\label{sec:aut}

In this preliminary section we are going to apply Theorem \ref{teo:nonloconst} to get nonlocal constants and related first integrals
for the \emph{autonomous} Jacobi equation
\begin{equation}\label{eq:reom}
\ddx+\tfrac{1}{2}\vp'\dx^2+B=0\,,
\end{equation}
where $\vp=\vp(x)$ and $B=B(x)$ do not explicitly depend on the time variable $t$. Here, the prime symbol stands for total derivatives with respect to $x$. 

As anticipated in the introduction, the general Jacobi-type form \eqref{eq:J} can be derived as the Euler-Lagrange equation \eqref{eq:E-L} for a Lagrangian \eqref{eq:Nucci}.  Therefore, the autonomous Lagrangian for \eqref{eq:reom} can be expressed as 
\begin{equation}\label{eq:rL}
\mathcal{L}(x,\dx)=\tfrac{1}{2}e^{\vp}\dx^2+\delta_2\,,
\end{equation}
for some $\delta_2=\delta_2(x)$ such that
\begin{equation}\label{eq:constrL}
\delta_2'=-e^{\vp}B\,.
\end{equation}
Hence, writing down the nonlocal constant \eqref{eq:nonlocGZ} in terms of \eqref{eq:rL}, we have
\begin{equation}\label{eq:nn}
I=e^\vp \dx \cdot \partial_\ve x_\ve\big\rvert_{\ve=0}
-\int_{t_0}^t \Big\{\big(\tfrac{1}{2}\vp' e^\vp\dx^2+\delta_2'\big)\cdot \partial_\ve x_\ve\big\rvert_{\ve=0}+e^\vp \dx\cdot \partial_\ve \dx_\ve\big\rvert_{\ve=0}\Big\}\, ds\,,
\end{equation}
where the perturbed family $x_\ve$ has not yet been included. 
Usefully, since $x_\ve$ is assumed to be a smooth function, we can rewrite $\partial_\ve \dx^\pm_\ve\rvert_{\ve=0}=(\partial_\ve x^\pm_\ve\rvert_{\ve=0})^{\dt}$. 

In literature, some results about autonomous Painlev\'{e}-Gambier equations present constants that are nonlocal due to the presence of an integral in the argument of an exponential factor \cite{GUHA20103247}. Inspired by this fact, it is quite natural for us to consider as perturbed  family the $x$-shift parametrization
\begin{equation}\label{eq:pert}
x_\ve^\pm=x+\ve a\etb\,,
\end{equation}
where $a=a(x)$ and $b=b(x)$ are free functions. The $\pm$ formalism indicates that we are simultaneously considering both families with a positive and a negative sign in the argument of the exponential.
With the family \eqref{eq:pert} and the constraint \eqref{eq:constrL}, our nonlocal constant  \eqref{eq:nn} can be simplified to be
\begin{equation}\label{eq:I1}
I_\pm=a\dx e^{\vp}\etb-\int_{t_0}^t\!\Big\{\!\left(\tfrac{1}{2}a\vp'+a'\right)e^{\vp}\dx^2\pm a b e^{\vp}\dx-aBe^{\vp}\Big\}\,\esb ds\,.
\end{equation}

Formally, expression \eqref{eq:I1} exhibits a double integration that we would like to remove. Furthermore we would also like to avoid the dependence of integrands on $\dx$. A possible approach in this direction is to (i) cancel the term proportional to $\dx^2$ by imposing $\tfrac{1}{2}a\vp'+a'=0$, which gives
\begin{equation}\label{eq:a}
a=e^{-\vp/2}\,,
\end{equation}
and (ii) reformulate  the integrand in $ds$ as a total derivative.

To establish request (ii), let us  assume that $B\equiv b\B$ for some $\B=\B(x)$. 
Using the above definition, expression \eqref{eq:I1} reduces to
\begin{equation}\label{eq:I2}
I_\pm=a\dx e^{\vp}\etb-\int_{t_0}^t\!\Big\{\!(\pm a b \dx e^{\vp})\,\esb+(\mp a\B e^{\vp})\big(\esb\big)^{\dt}\Big\}\, ds\,.
\end{equation}
Thus, assuming $\pm a b \dx e^{\vp}=(\mp a\B e^{\vp})^{\dt}$ as constraint, the integrand of expression \eqref{eq:I2} becomes a total derivative as desired
\begin{equation}\label{eq:I3}
I_\pm=a\dx e^{\vp}\etb\pm\int_{t_0}^t\!\big(a\B e^{\vp}\esb\big)^{\dt}ds\,.
\end{equation}
It is not difficult to verify that the above constraint can be rewritten as
\begin{equation}\label{eq:b}
b=-\tfrac{1}{2}\vp'\B-\B'\,.
\end{equation}

According to \eqref{eq:I3}, we can evaluate the integral in the second term of the expression up to an additive integration constant that we assume to be zero.
Then, one can  use \eqref{eq:a} and \eqref{eq:b} to further simplify the result  as follows
\begin{equation}\label{eq:I4}
I_\pm=(\dx\pm\B)\,e^{\vp/2}e^{\mp\!\int_t(\vp'\!\B/2+\B')d\tau}\,.
\end{equation}

Note that expression \eqref{eq:I4} is a nonlocal constant for the autonomous Jacobi equation \eqref{eq:reom} that still depends on $\B$. 
By construction, equation \eqref{eq:b} constraints $\B$ as a function of $b$ and $\vp$. 
So, multiplying equation \eqref{eq:b} by $\B$ and defining $y\equiv\B^2$, we get $y'+\vp'y+2B=0$.  This is a first-order linear ODE, whose solution is
\begin{equation}
\label{eq:rBbar}
\B=\sqrt{-2e^{-\vp}\smallint_xe^\vp B\,dx}
=\sqrt{2\delta_2 e^{-\vp}}\,.
\end{equation}

\begin{theorem}\label{teo:xxx}
Let $e^{\vp}B=-\delta_2'$ for some $\delta_2$. Then, $\mathcal{I}$ is a first integral for the autonomous Jacobi equation $\ddx+\tfrac{1}{2}\vp'\dx^2+B=0$ with
\begin{equation}\label{eq:I}
\mathcal{I}=\tfrac{1}{2}\dx^2e^{\vp}-\delta_2\,.
\end{equation}
\end{theorem} 
\begin{proof}
Use expression \eqref{eq:I4} to compute $\mathcal{I}\equiv\tfrac{1}{2} I_+I_-=\tfrac{1}{2}(\dx^2-\B^2)\,e^{\vp}$. Then, substitute expression \eqref{eq:rBbar} inside $\mathcal{I}$.
\end{proof}


\subsection{Application: autonomous Painlev\'{e}-Gambier equations}
\label{sec:oiu}

There are several equations of the \emph{Painlev\'{e}-Gambier} classification that belong to the autonomous  parametrization \eqref{eq:reom} (see, e.g. Refs. \cite{GUHA20103247,Dambrosi2021}).  

\begin{example}

As we show in more details below, the constants we found in the previous section are useful to study some  Painlev\'{e}-Gambier equations.
The Painlev\'{e}-Gambier equations {\footnotesize XVIII}, {\footnotesize XXI}, and {\footnotesize XXII} can be parametrized as
\begin{equation}\label{eq:Painleve}
\ddx-\tfrac{1}{2}\alpha x^{-1}\dx^2+\beta x^n=0\,,
\end{equation}
with $\{\alpha,\beta,n\}\equiv\chi$ sets of constant parameters
\begin{equation}\label{eq:coeffsPain}
\chi_{\mbox{\tiny XVIII}}\!=\!\big\{1,\!-4,2\big\}\quad
\chi_{\mbox{\tiny XXI}}\!=\!\big\{\tfrac{3}{2},\!-3,2\big\}\quad
\chi_{\mbox{\tiny XXII}}\!=\!\big\{\tfrac{3}{2},1,0\big\}.
\end{equation}

First of all, a comparison between equation \eqref{eq:reom} and equation \eqref{eq:Painleve} yields  $\vp=-\alpha\ln x$ and $B=\beta x^n$. Then, following the definition \eqref{eq:rBbar}, we deduce that $\B=\sqrt{\beta \gamma}\, x^{(n+1)/2}$ with $\gamma= 2/(\alpha-n-1)$. 

Hence, after some simple calculations,  our nonlocal constants \eqref{eq:I4} 
become
\begin{equation}\label{eq:r1}
I_\pm=\big\{\dx x^{-\alpha/2}\pm \sqrt{\beta \gamma}\, x^{-1/\gamma}\big\}\,e^{\mp\sqrt{\beta/\gamma}\int_tx^{(n-1)/2}d\tau}\,,
\end{equation}
Correspondingly, the first integral \eqref{eq:I} is obtained to be
\begin{equation}\label{eq:r2}
\mathcal{I}=\dx^2x^{-\alpha}-\beta \gamma x^{-2/\gamma}\,.
\end{equation}

We mention that if we restrict to the $\mathcal{I}=0$ hypersurface, equation \eqref{eq:r2} can be simply solved. In particular, after integration we get
\begin{equation}\label{eq:xkre}
x(t) = A\left(t+c_1\right)^{\sigma},
\end{equation}
where $A\equiv(\sqrt{\beta\gamma}/\sigma)^\sigma$ and $c_1\equiv c/\sqrt{\beta\gamma}$, $c\in\mathbb{R}$,
with $\sigma\equiv \left(1/\gamma-\alpha/2+1\right)^{-1}$. 

When combined with coefficients \eqref{eq:coeffsPain}, our results \eqref{eq:r1}, \eqref{eq:r2} and \eqref{eq:xkre} match exactly with some constants and solutions proposed case-by-case in Refs.  \cite{GUHA20103247, GUHA2022} using generalized Sundman transformations.  The advantage of our approach lies in having identified a single formula for all of them.

\end{example}



\section{Non-autonomous Jacobi equation}

The procedure we discussed in the previous section have shown us the crucial ingredients to apply Theorem \ref{teo:nonloconst} within a simplified case of Jacobi equation.
Given the way to proceed, in the rest of this paper we turn to the large class of \emph{non-autonomous} Jacobi equation 
\begin{equation}\label{eq:eom}
\ddx+\tfrac{1}{2}\partial_x\vp\,\dx^2+\partial_t\vp \,\dx+B=0\,,
\end{equation}
whose parameters $\vp=\vp(t,x)$ and $B=B(t,x)$ explicitly depend on the time variable $t$. As anticipated in the introduction, this class of equations can be derived as the Euler-Lagrange equation \eqref{eq:E-L} for a Lagrangian 
\begin{equation}\label{eq:L}
L(t,x,\dx)\equiv\tfrac{1}{2}e^{\vp}\dx^2+\delta_1\dx+\delta_2\,,
\end{equation}
with $\delta_1=\delta_1(t,x)$ and $\delta_2=\delta_2(t,x)$ satisfying
\begin{equation}\label{eq:constrLL}
\partial_t\delta_1-\partial_x\delta_2=e^\vp B\,.
\end{equation}

With the above discussions in mind, 
let us begin evaluating expression \eqref{eq:nonlocGZ} on  Lagrangian \eqref{eq:L}. Expanding the derivatives of $L$ with respect to $x$ and $\dx$ we have a nonlocal constant that looks like this
\begin{equation}\label{eq:xx}
I=\big(e^\vp \dx+\delta_1\big)\cdot \partial_\ve x_\ve\big\rvert_{\ve=0}
-\int_{t_0}^t \xi\, ds\,,
\end{equation}
in which the integrand $\xi=\xi(s,x(s))$ is
\begin{equation}\label{eq:XII}
\xi \equiv \big(\tfrac{1}{2}\p{x}\vp \,e^\vp\dx^2+\p{x}\delta_1\dx+\p{x}\delta_2\big)\cdot \partial_\ve x_\ve\big\rvert_{\ve=0}\,+\big(e^\vp \dx+\delta_1\big)\cdot \partial_\ve \dx_\ve\big\rvert_{\ve=0}\,.
\end{equation}

Taking inspiration from the autonomous case, we choose the $x$-shift family
\begin{equation}\label{eq:pertx}
x_\ve^\pm=x+\ve a\etbb\,,
\end{equation}
as perturbed motion, where $a=a(t,x)$ and $b_\pm=b_\pm(t,x)$ are  free functions. 

Using expression \eqref{eq:pertx}, as well as the property $\partial_\ve \dx^\pm_\ve\rvert_{\ve=0}=(\partial_\ve x^\pm_\ve\rvert_{\ve=0})^{\dt}$ with the condition \eqref{eq:constrLL}, one finds that the integrand \eqref{eq:XII} can be expressed as
\begin{multline}\label{eq:mkj}
\xi_\pm=
\big(a\delta_1\esbb\big)^{\dt}+
(\tfrac{1}{2}a\p{x}\vp+\p{x}a)e^\vp\dx^2\esbb\\
+\big\{\big[(\p{t}a\pm a b_\pm)\,e^\vp \dx\big]\esbb-(a e^\vp B)\,\esbb\big\}\,.
\end{multline}

Compared with the autonomous case, the above expression looks more problematic to be integrated, as it now (i) explicitly depends on time, and (ii) exhibits new additional contributions. However, things become simplified by noting that various terms group themselves into total derivatives if we choose appropriately $a$ and $b_\pm$.

To achieve this, the structure of expression \eqref{eq:mkj} suggests first of all  to neglect the term proportional to $\dx^2$ by imposing $\tfrac{1}{2}a\p{x}\vp+\p{x}a=0$, which yields
\begin{equation}\label{eq:at}
a=e^{-\vp/2}\,.
\end{equation}
Going on, instead of working directly with $b_\pm$ it is more convenient to introduce a new variable $\B_\pm=\B_\pm(t,x)$ such that $B\equiv b_\pm\B_\pm$. Expression \eqref{eq:mkj} thus becomes
\begin{equation}\label{eq:kll}
\xi_\pm\!=\!
\big(a\delta_1\esbb\big)^{\dt}
\!+\!\big\{\big[(\p{t}a\pm a b_\pm)\,e^\vp \dx\big]\esbb\!+\!(\mp a e^\vp \B_\pm)\,\big(\esbb\!\big)^{\dt}\big\}.
\end{equation}
Then, if we further require that $b_\pm$ satisfies
\begin{equation}\label{eq:dd}
(\p{t}a\pm a b_\pm)\,e^\vp \dx=(\mp a e^\vp \B_\pm)^{\dt}\,,
\end{equation}
it is immediate to check that $\xi_\pm$ becomes a total derivative as anticipated before
\begin{equation}\label{eq:zxc}
\xi_\pm=\big[a\delta_1\esbb\mp(ae^\vp\B_\pm)\esbb\big]^{\dt}\,.
\end{equation}

According to \eqref{eq:zxc}, we can evaluate the integral in the second term of expression \eqref{eq:xx} up to an additive integration constant that we assume to be zero. More specifically, we get
\begin{equation}\label{eq:ddd}
I_\pm=(\dx\pm\B_\pm)\,a e^\vp\etbb\,,
\end{equation}

This result depends on $b_\pm$ and $\B_\pm$, whose explicit expressions are unknown so far. To reveal their form, let us consider condition \eqref{eq:dd} and expand it as follows
\begin{equation}\label{eq:lklk}
\mp(\p{t}a\pm a b_\pm)e^\vp \dx = \p{t}(a e^\vp \B_\pm)+\p{x}(a e^\vp \B_\pm)\dx\,.
\end{equation} 

Since equation \eqref{eq:lklk} must be true for all solutions, it can be decomposed with respect $\dx$. In this way, such equation can be easily reformulated as follows
\begin{equation}
\label{eq:klk}
\begin{cases}
\p{t}(a e^\vp \B_\pm)=0\,,\\
\p{x}(a e^\vp \B_\pm)=\mp(\p{t}a\pm a b_\pm)e^\vp\,.
\end{cases}
\end{equation}
It is immediate to check that the first equation of \eqref{eq:klk} is satisfied by 
\begin{equation}\label{eq:k}
k_\pm\equiv a e^\vp \B_\pm\,,
\end{equation}
with $k_\pm=k_\pm(x)$ a time-independent function. 
On the other hand, using the mixed partials equality $\p{x}\p{t}k_\pm=\p{t}\p{x}k_\pm$, the second equation of \eqref{eq:klk} provides the relation
\begin{equation}
\p{t}\big[(\p{t}a\pm a b_\pm)e^\vp\big]=0\,,
\end{equation}
that can be expanded to be a first-order (linear) differential equation for $b_\pm$ as follows
\begin{equation}\label{eq:ssdf}
a\,\p{t}b_\pm+(\p{t}\vp\, a+\p{t}a)\,b_\pm\pm(\p{t}\vp\p{t}a+\p{t}\p{t}a)=0\,.
\end{equation}

Using standard techniques we solve equation \eqref{eq:ssdf} and impose condition \eqref{eq:klk} to fix the integration constant. Hence, after several calculations, we finally obtain
\begin{equation}\label{eq:fgh}
b_\pm=\mp\,\p{t}a-e^{-\vp}a^{-1}\p{x}k_\pm\,.
\end{equation}

Since $B\equiv b_\pm\B_\pm$, we multiply expression \eqref{eq:fgh} by $\B$ and use definition \eqref{eq:k} to substitute $k_\pm$. Consequently, our expression simplifies to
\begin{equation}\label{eq:kjl}
\B_\pm\p{x}\B_\pm+\p{x}(\ln a+\vp)\B_\pm^2\pm(\p{t}\ln a) \B_\pm=-B\,,
\end{equation}
that finally constraints $\B_\pm$ in terms of the original functions $B$ and $\vp$ only. Once solved, the solution $B_\pm$ will automatically  fix $b_\pm(B,\vp)=B/ \B_\pm(B,\vp)$ by definition. 

In contrast with the autonomous case, equation \eqref{eq:kjl} is a nonlinear differential equation  for $\B_\pm$. 
Hence, without loss of generality, we choose to proceed by splitting our study depending on wether $\p{t}\vp\sim\,\p{t}\ln a$ is equal to zero or not.
  
  
\subsection{The $\p{t}\vp=0$ case}
\label{eq:pt0}

If $\p{t}\vp=0$, equation \eqref{eq:kjl} becomes
\begin{equation}\label{eq:ghj}
\B_\pm\p{x}\B_\pm+\p{x}(\ln a+\vp)\B_\pm^2=-B\,.
\end{equation}
As in the autonomous case, substituting definition \eqref{eq:at} in equation \eqref{eq:ghj} and assuming $y_\pm\equiv\B_\pm^2$, we get 
\begin{equation}
\p{x} y_\pm+(\p{x}\vp) y_\pm+2B=0\,.
\end{equation}
This is a first-order linear ODE, whose solution is
\begin{equation}\label{eq:rBbare}
\B=\sqrt{-2e^{-\vp}\smallint^x_{x_0}(e^\vp B)\rvert_{t}\,dx}=
\sqrt{2e^{-\vp}\smallint^x_{x_0}(\p{x}\delta_2-\p{t}\delta_1)\rvert_{t}\,dx}\,.
\end{equation}

On the other hand, being $\p{t}\vp=0$, equation \eqref{eq:fgh} can be replaced by
\begin{equation}\label{eq:lkiqw}
b_\pm=e^{-\vp}a^{-1}\p{x}k_\pm\,.
\end{equation}
that, bearing in mind definition \eqref{eq:k} and expression \eqref{eq:rBbare}, finally yields
\begin{equation}\label{eq:lkop}
b\equiv b_\pm=-\tfrac{1}{2}\p{x}\vp-\p{x}\B\,.
\end{equation}
Note that in this case the perturbed families $x_\ve^\pm$ return the same value for $b_\pm$.

We have thus obtained all the crucial ingredients to evaluate our nonlocal constant \eqref{eq:ddd} and, more specifically, we get
\begin{equation}\label{eq:ddddsd}.
I_\pm=(\dx\pm\B)\, e^{\vp/2}e^{\mp\int_t(\p{x}\vp/2+\p{x}\B)d\tau}\,.
\end{equation}
with $\B$ given by \eqref{eq:rBbare}.

\begin{theorem}\label{teo:xxffx}
Let $\lambda = -\int e^\vp B \,dx$ with  $\p{t}\vp=0$. Then, $I$ is  constant along the solutions of the Jacobi equation $\ddx+\tfrac{1}{2}\partial_x\vp\,\dx^2+B=0$ with
\begin{equation}\label{eq:Il}
I\equiv\tfrac{1}{2}\dx^2e^{\vp}+\lambda-\int^t_{t_0}\p{t}\lambda\,dt\,.
\end{equation}
\end{theorem} 
\begin{proof}
Use expression \eqref{eq:ddddsd} to calculate $I\equiv\tfrac{1}{2} I_+I_-=\tfrac{1}{2}(\dx^2-\B^2)\,e^{\vp}$. Define $\delta_1\equiv \p{x}\eta$ and consider expression \eqref{eq:rBbare}, that can be rewritten as 
\begin{eqnarray}
\B^2&=&2e^{-\vp}\smallint^x_{x_0}\p{x}\big(\delta_2-\p{t}\eta\big)\rvert_{t}\,dx\,,\nonumber\\
&=&2e^{-\vp}\big[\delta_2-\p{t}\eta-\smallint^t_{t_0}\p{t}\big(\delta_2-\p{t}\eta\big)\,dt\big]\,,\nonumber \\
\label{eq:lol}
&=&2e^{-\vp}\big(\lambda-\smallint^t_{t_0}\p{t}\lambda\,dt\big)\,.
\end{eqnarray}
where we defined $\lambda \equiv\delta_2-\partial_t\eta$.
Finally, substitute expression \eqref{eq:lol} inside $I$. It follows from \eqref{eq:constrLL} that $e^\vp B=-\p{x}\lambda$.
\end{proof}

Notice that, as expected, the integrand of \eqref{eq:Il} depends only on $x$ and $t$. In addition, as we will see in the following Remark \ref{rem:1}, expression \eqref{eq:Il} exactly recovers our previous result \eqref{eq:I} when evaluated in the autonomous case.

\begin{remark}\label{rem:1}
One might wonder under which conditions expression \eqref{eq:Il} becomes a true first integral, in the sense of a local function of $t$. For this purpose, it is sufficient to require that $\p{t}\lambda=\dot\psi$ for some $\psi=\psi(t)$. Derive such condition with respect to $x$, to get $0=\p{x}\p{t}\lambda$. Switch the partial derivatives, and use the hypothesis $\p{x}\lambda=-e^\vp B$ of Theorem \ref{teo:xxffx} to finally write
$0=\p{t}(e^\vp B)=e^\vp \p{t}B$, which implies $\p{t}B=0$. Since in this subsection we also assume $\p{t}\vp=0$, expression \eqref{eq:Il} returns a true first integral only if we move to the autonomous case.
\end{remark}


\subsubsection{Application: non-autonomous Painlev\'{e}-Gambier equations}
\label{sec:lkjhg}

There are several equations of the \emph{Painlev\'{e}-Gambier} classification which belong to the Jacobi parametrization \eqref{eq:eom} (see, e.g. Refs. \cite{GUHA20103247,Dambrosi2021}).  

\begin{example}

The Painlev\'{e}-Gambier equation {\small IV} is
\begin{equation}\label{eq:Painlevesd}
\ddx-(6x^2+t)=0\,.
\end{equation}

Let us first establish a comparison between equation \eqref{eq:Painlevesd} and Theorem \ref{teo:xxffx} to infer that $\vp=0$ and $B=-(6x^2+t)$. Since $e^\vp B=-(6x^2+t)$, we have $\lambda=-x(2x^2+t)$. Then, expression \eqref{eq:Il} provides the nonlocal  constant
\begin{equation}
I=\tfrac{1}{2}\dx^2-2x^3-xt+\smallint_{t_0}^t x\,dt\,.
\end{equation}

\end{example}

\begin{example}

The Painlev\'{e}-Gambier equation {\small XX} is
\begin{equation}\label{eq:Painlevesgfgfd}
\ddx-\tfrac{1}{2}x^{-1}\dx^2-4x^2-2tx=0\,.
\end{equation}

A comparison between equation \eqref{eq:Painlevesgfgfd} and Theorem \ref{teo:xxffx} yields $\vp=-\ln x$ and $B=-2x(2x+t)$. Since $e^\vp B=-2(2x+t)$, we have $\lambda=-2x(x+t)$. Then, expression \eqref{eq:Il} yields
\begin{equation}
I=\tfrac{1}{2}\dx^2x^{-1}-2x(x+t)+2\smallint_{t_0}^t x\,dt\,.
\end{equation}

\end{example}



\subsection{The $\p{t}\vp\ne0$ case}
\label{eq:lowe}

When $\p{t}\vp\ne0$, the mathematical structure of equation \eqref{eq:kjl} is more complicated to be analyzed.
We start by multiplying equation \eqref{eq:kjl} by $e^{2(\ln a+\vp)}$ and using definition \eqref{eq:k} to write
\begin{equation}\label{eq:lki}
\p{x}k_\pm^2\pm2\p{t}a e^\vp k_\pm=-2 B a^2 e^{2\vp}\,.
\end{equation}

We take the $\partial_t$ derivative of equation \eqref{eq:lki} and switch $\p{t}\p{x}k^2_\pm=\p{x}\p{t}k^2_\pm$. Consequently, the above equation can be then written as
\begin{equation}\label{eq:lopw}
\p{x}\p{t}k^2_\pm\pm2\p{t}(\p{t}a e^\vp) k_\pm\pm2\p{t}a e^\vp (\p{t} k_\pm)=-2\p{t}(B a^2 e^{2\vp})\,.
\end{equation}

Furthermore, since $\p{t}k^2_\pm=2k_\pm\p{t}k_\pm=0$ (remember that expression \eqref{eq:klk} fixes $\p{t}k_\pm=0$),  from equation \eqref{eq:lopw} it follows that 
\begin{equation}\label{eq:swe}
\p{t}(\p{t}a\, e^\vp) k_\pm=\mp\p{t}(B a^2 e^{2\vp})\,,
\end{equation}
that can be solved with respect to $k_\mp$. More precisely, we obtain
\begin{equation}\label{eq:klol}
k_\pm=\mp\frac{\p{t}(B a^2 e^{2\vp})}{\p{t}(\p{t}a\, e^\vp)}\,.
\end{equation}
which must satisfy the set of conditions \eqref{eq:klk}.

An immediate consequence of the above result is the evaluation of $\B$. In fact, by plugging \eqref{eq:klol} and \eqref{eq:at} in the right hand side of \eqref{eq:k} we obtain
\begin{equation}\label{eq:lopopo}
\B_\pm=\pm4\frac{\p{t}B+B\p{t}\vp}{2\p{t}\p{t}\vp+(\p{t}\vp)^2}\,.
\end{equation}
On the other hand,  we can use expression \eqref{eq:lopopo} to calculate $b_\pm=B/\B_\pm$. Hence,
\begin{equation}\label{eq:lnddb}
b_\pm=\pm\frac{1}{4}\frac{2\p{t}\p{t}\vp+(\p{t}\vp)^2}{\p{t}\ln B+\p{t}\vp}\,.
\end{equation}

\begin{theorem}\label{teo:xxkjkffx}
Let $\p{t}\vp\ne0$ and $Be^\vp=\rho_1 e^{\vp/2}+\rho_2$ for some $\rho_{1,2}=\rho_{1,2}(x)$ such that
\begin{equation}\label{eq:rr}
\rho_1'=\p{t}(\ln\p{t}\vp)\big(e^{\vp/2}+\rho_2/\rho_1\big)\,.
\end{equation}
Then, $I$ is constant along the solutions of the Jacobi equation $\ddx+\tfrac{1}{2}\partial_x\vp\,\dx^2+\partial_t\vp \,\dx+B=0$ with
\begin{equation}\label{eq:lkj}
\!I\!=\!\left(\!\dx e^{\vp/2}\!+\!\frac{2\rho_1}{\p{t}(2\ln \p{t}\vp\!+\!\vp)}\!\right)\!\exp\!\left\{\frac{1}{2}\!\int_t 
\p{t}\big(2\ln \p{t}\vp\!+\!\vp\big)\!\left(\!1\!+\!\frac{\rho_2}{\rho_1}e^{-\vp/2}\!\right)\!d\tau\right\}\!.
\end{equation}
\end{theorem}
\begin{proof}
We begin substituting \eqref{eq:at} in \eqref{eq:klol} and imposing $\p{t}{k_\pm}=0$ (see the first condition of \eqref{eq:klk}). 
As a consequence, we get $\p{t}{(B e^\vp})=\rho_1\p{t}e^{\vp/2}$, whose solution is $B=\rho_1 e^{-\vp/2}+\rho_2 e^{-\vp}$ (our first hypothesis). We plug $B$ inside \eqref{eq:lopopo} and \eqref{eq:lnddb} to evaluate $\B_\pm$ and $b_\pm$, that we use with definition \eqref{eq:ddd} to get our final result
\eqref{eq:lkj}. Note that $I_+=I_-$, hence $I=I_\pm$. Since $k_\pm=\pm\rho_1$, our second hypothesis \eqref{eq:rr} is obtained from $\p{x}k_\pm=0$ (the second condition of \eqref{eq:klk}).
\end{proof}
Notice that also in this case the integrand of \eqref{eq:lkj} depends only on $x$ and $t$.

\begin{remark}\label{rem:df}
One might wonder under which conditions expression \eqref{eq:lkj} becomes a true first integral, in the sense of a local function of $t$. For this purpose, since the integrand depends only on $t$ and $x$, it is sufficient to require that $\p{t}(2\ln \p{t}\vp\!+\!\vp)\!\left(1\!+\!\rho_2/\rho_1e^{-\vp/2}\right)=\dot \psi(t)$ for some $\psi=\psi(t)$.
\end{remark}


\subsubsection{Application: non-autonomous Jacobi equation}
\label{eq:koprwrt}

\begin{example}
\label{ex:fd}

Consider the following Jacobi-type equation
\begin{equation}\label{eq:jaerctot}
\ddx+\tfrac{1}{2}\dx^2+\dx+\varrho\, e^{-(t+x)/2}=0\,,
\end{equation}
with $\varrho$ a constant parameter. Since $B=\varrho\, e^{-(t+x)/2}$ and $\p{x}\vp=\p{t}\vp=1$, we have $\vp=t+x$. Hence, equation \eqref{eq:jaerctot} satisfies the hypotheses of Theorem \ref{teo:xxkjkffx} with $\rho_1=\varrho$ and $\rho_2=0$.  Note that Remark \ref{rem:df} is also satisfied by $\psi=t$. Hence, expression \eqref{eq:lkj} yields the following true first integral
\begin{equation}\label{eq:lpo}
\mathcal{I}=e^{\,t/2}\big\{\dx e^{(t+x)/2}+2\varrho\big\}.
\end{equation}

Being linear in $\dx$, $\mathcal{I}$ turns out to be a precious tool to find a solution of equation \eqref{eq:jaerctot}. In fact, expression \eqref{eq:lpo} can be rewritten as $\dot{\mathcal{J}}(t,x)=0$, with
\begin{equation}\label{eq:kloi}
\mathcal{J}= 2 \,e^{\,x/2}+\tilde{\mathcal{I}} e^{-t}-4\varrho \,e^{-t/2}\,.
\end{equation}
Hence, $\mathcal{J}$ is a first integral too. After some calculations, it can be easily checked that expression \eqref{eq:kloi} solves directly equation \eqref{eq:jaerctot} with
\begin{equation}
x(t)=2\ln\big\{2\varrho \,e^{-t/2}-\tfrac{1}{2}\,\tilde{\mathcal{I}}\, e^{-t}+\tfrac{1}{2}\tilde{\mathcal{J}}\big\}\,.
\end{equation}
Here, the $\tilde{\mathcal{I}}$ and $\tilde{\mathcal{J}}$ parameters are constants.

\end{example}

\begin{example}
An interesting application of Theorem \ref{teo:xxkjkffx} arises when $\rho_2(x)=0$. Under such condition, the hypotheses  of our theorem become
\begin{equation}\label{eq:klo}
\begin{cases}
\rho_1'(x)=\p{t}\!\ln\p{t}\vp(t,x)\,e^{\vp(t,x)/2},\\
B(t,x)e^{\vp(t,x)}=\rho_1(x)\, e^{\vp(t,x)/2}\,.
\end{cases}
\end{equation}
From the above constraints, the first differential equation can be solved easily. More precisely, we get
\begin{equation}\label{eq:loeprt}
\vp(t,x)=2\ln\left\{\frac{2\rho_1'(x)+e^{\frac{1}{2}c_1(x)\left[t+c_2(x)\right]}}{c_1(x)}\right\}.
\end{equation}
with $c_1(x)\ne 0$ and $c_2(x)$ integration functions.

On the other hand, using expression \eqref{eq:loeprt}, the second equation of \eqref{eq:klo} gives
\begin{equation}
\label{eq:bpo}
B(t,x)=\frac{c_1(x)\rho_1(x)}{2\rho_1'(x)+e^{\frac{1}{2}c_1(x)[t+c_2(x)]}}\,.
\end{equation}
With the above formulas in mind, we see that once fixed $c_1(x)$, $c_2(x)$ and $\rho_1(x)$ we can deduce the class of Jacobi equations \eqref{eq:eom} characterized by a first integral like \eqref{eq:lkj} with $\rho_2(x)=0$. This is the result of the following proposition.

\begin{proposition}
Let $c_{1,2,3}=c_{1,2,3}(x)$ be three free functions of $x$ and define $\vartheta(t,x)\equiv c_1 (c_2 + t)$. Then,  $\mathcal{I}$ is a first integral for the equation
\begin{equation}
\label{eq:loertyu}
\ddx+
\frac{4 c_1 c_3'' \!-\! 4 c_3' c_1' \!+\! e^{\frac{1}{2} \vartheta}\big[c_1' (\vartheta\!-\!2) \!+\! c_1^2 c_2'\big]}{2c_1 \big(2 c_3'\!+\!e^{\frac{1}{2} \vartheta}\big)}\,\dx^2
+\frac{c_1 e^{\frac{1}{2} \vartheta}}{2 c_3'\!+\!e^{\frac{1}{2} \vartheta}}\,\dx
+\frac{c_1c_3}{2c_3'\!+\!e^{\frac{1}{2}\vartheta}}
=0\,,
\end{equation}
with
\begin{equation}\label{eq:LPq}
\mathcal{I}=e^{\frac{1}{2}c_1t}\Big[2\big(c_3'+c_3\big)+\dx e^{\frac{1}{2}\vartheta}\Big]c_1^{-1}\,.
\end{equation}
\begin{proof}
Let us define $\vartheta\equiv2\ln\p{t}\vp+\vp$. By choosing expression \eqref{eq:loeprt} in the definition above it turns out that $\vartheta=c_1(t+c_2)$. Hence, equation \eqref{eq:loertyu} and its first integral \eqref{eq:LPq} are a direct consequence of \eqref{eq:eom} and \eqref{eq:lkj}, where $\vp$ and $B$ have the form \eqref{eq:loeprt} and \eqref{eq:bpo}. Here, we assumed that $\rho_2=0$ and we defined $c_3\equiv\rho_1$.
\end{proof}
\end{proposition}

Let $\bar{c}_{1,2,3} \in\mathbb{R}$ and $\omega_{1,2}\in\mathbb{R}$ be some constant parameters. Interestingly, if we take $c_1=\bc{1}$, $c_2=\omega_1\bc{1}^{-1}x+\bc{2}$ and $c_3=\omega_2\bc{1}+\bc{3}e^{-x}$, it is possible to recast expression \eqref{eq:LPq} to look as the conservation law of one more first integral. Indeed, taking the antiderivative in time of \eqref{eq:LPq} we obtain
\begin{equation}
\label{eq:lpqert}
\mathcal{K}=2\omega_1^{-1}e^{\frac{1}{2}(\omega_1 x+\bc{1})}+\tilde{\mathcal{I}}e^{-\bc{1}t}-4\omega_2 e^{-\frac{1}{2}\bc{1}t}\,,
\end{equation}
with $\dot{\mathcal{K}}(t,x)=0$. Expression \eqref{eq:lpqert} is equivalent to a Cauchy problem for a first-order differential equation with separated variables. It is easy to see that the general solution of this equation is 
\begin{equation}
\label{eq:solr}
x(t)=2\omega_1^{-1}\ln\left\{\omega_1 e^{-\frac{1}{2}\bc{1}\bc{2}}\left[2\omega_2 e^{-\frac{1}{2}\bc{1}t}-\tfrac{1}{2}\tilde{\mathcal{I}}e^{-\bc{1}t}+\tfrac{1}{2}\tilde{\mathcal{K}}\right]\right\}\,.
\end{equation}
Here, the $\tilde{\mathcal{I}}$ and $\tilde{\mathcal{K}}$ parameters are constants.

The class of equations \eqref{eq:loertyu} includes as a particular case the equation \eqref{eq:jaerctot} of Example \ref{ex:fd}, where $c_1=1$, $c_2=x$ and $c_3=\varrho$. Such case also corresponds to our results \eqref{eq:lpqert} and \eqref{eq:solr}
with $\bc{1}=\omega_1=1$, $\bc{2}=\bc{3}=0$ and $\omega_2=\varrho$.
\end{example}


\section{Acknowledgments}
\label{sec:ak}
The author would like to thank Professor Gaetano Zampieri for useful discussions.


\nocite{*}
\providecommand{\bysame}{\leavevmode\hbox to3em{\hrulefill}\thinspace}


\end{document}